\documentclass[%
 reprint,
 amsmath,amssymb,
 aps,
]{revtex4-2}

\usepackage{graphicx}
\usepackage{dcolumn}
\usepackage{bm}
\usepackage{color}
\usepackage{amsfonts}
\usepackage{dsfont}
\usepackage{amsmath}
\usepackage{amssymb}
\usepackage{changes}
\usepackage{comment}

\newtheorem{corollary}{Corollary}

\newtheorem{proposition}{Proposition}

\newenvironment{proof}[1][Proof]{\noindent\textbf{#1.} }{\ \rule{0.5em}{0.5em}}

 \newcommand{\ket}[1]{|#1\rangle}
 \newcommand{\bra}[1]{\langle #1|}
 
 \newcommand{\proj}[1]{\ket{#1}\bra{#1}}

\newcommand{\eye}{\mathds{1}}

\newcommand{\Tr}{\text{Tr}}

\begin{document}


\title{Measurement-device-independent entanglement detection \\ for continuous-variable systems}

\author{Paolo Abiuso}\email{paolo.abiuso@icfo.eu}
\affiliation{ICFO – Institut de Ciencies Fotoniques, The Barcelona Institute of Science and Technology, 08860 Castelldefels (Barcelona), Spain}
\author{Stefan B\"auml}
\affiliation{ICFO – Institut de Ciencies Fotoniques, The Barcelona Institute of Science and Technology, 08860 Castelldefels (Barcelona), Spain}

\author{Daniel Cavalcanti}
\affiliation{ICFO – Institut de Ciencies Fotoniques, The Barcelona Institute of Science and Technology, 08860 Castelldefels (Barcelona), Spain}
 
\author{Antonio Ac\'in}
\affiliation{ICFO – Institut de Ciencies Fotoniques, The Barcelona Institute of Science and Technology, 08860 Castelldefels (Barcelona), Spain}
\affiliation{ICREA - Instituci\'o Catalana de Recerca i Estudis Avan\c{c}ats, Passeig Lluis Companys 23, 08010 Barcelona, Spain}

\date{\today}

\begin{abstract}
We study the detection of continuous-variable entanglement, for which most of the existing methods designed so far require a full specification of the devices, and we present protocols for entanglement detection in a scenario where the measurement devices are completely uncharacterised.
We first generalise, to the continuous variable regime, the seminal results by Buscemi [PRL 108, 200401 (2012)] and  Branciard \emph{et al.} [PRL 110, 060405 (2013)], showing that all entangled states can be detected in this scenario.  Most importantly, we then describe a practical protocol that allows for the measurement-device-independent certification of entanglement of all two-mode entangled Gaussian states. This protocol is feasible with current technology as it makes only use of standard optical setups such as coherent states and homodyne measurements. 
\end{abstract}

\maketitle

\section{Introduction}

Entanglement is the main resource for a broad range of applications in quantum information science, among which are quantum key distribution \cite{Ekert91}, quantum computation \cite{JL2003}, and quantum metrology \cite{Toth2014}. It is therefore crucial to develop methods to detect entanglement that are reliable and practical. The most common method to detect entanglement is given by entanglement witnesses~\cite{EntReview_2009}. However, to be reliable this technique requires a perfect implementation of the measurements. Indeed, small calibration errors can lead to false-positive detection of entanglement \citep{Rosset12,Moroder10}, which can be critical when using  the wrongly detected entangled state for quantum information purposes. A possible way of circumventing this problem is to move into the so called device-independent (DI) scenario~\cite{NLreview}. In this framework measurements do not need to be characterised, since  entanglement is detected through the violation of Bell inequalities, which only use the statistics provided by the experiment, without making any assumptions on the real implementation. The DI scenario is however  stringent from an experimental point of view, requiring low levels of noise and high detection efficiencies. This is why other approaches requiring intermediate level of trust on the devices have been developed. In particular, there exist methods that do not require any characterization of the measurement implemented for entanglement detection, known as measurement-device-independent (MDI)~\cite{buscemi2012all,branciard2013measurement}.

Here, we consider the problem of entanglement detection in continuous-variable (CV) systems, for which very little is known about methods not requiring a full characterization of measurement devices. 
A fully DI approach is complex because of the difficulty of finding useful Bell tests for continuous-variable states. For instance, in the Gaussian regime, which is the most feasible experimentally, DI entanglement detection is impossible because no Bell inequality can be violated~\cite{NLreview}, hence intermediate approaches are necessary. The main goal of this work is to provide methods for MDI entanglement detection in CV systems. We first demonstrate that, in principle, all entangled states can be detected in this scenario. Then, we describe a MDI protocol that can detect the entanglement of all two-mode Gaussian states. Our protocol only  relies on the use of trusted, well-calibrated, sources of coherent states, the easiest to prepare in the lab.

An entanglement detection scenario where two parties, Alice and Bob, do not assume a particular description of their measurement but use trusted sources of states was first introduced by Buscemi \cite{buscemi2012all}. Namely, let us consider that Alice and Bob can produce states $\psi_{A'}^\mu$ and $\psi_{B'}^{\nu}$ according to some distributions $P_{A'}(\mu)$, $P_{B'}(\nu)$ respectively. Alice and Bob can use these states as inputs to their measurement devices, which return outcomes $a$ and $b$ respectively. These outcomes occur with probability $P(a,b|\psi_{A'}^\mu,\psi_{B'}^{\nu})=\Tr[M_{AA'}^a\otimes N_{BB'}^b (\psi_{A'}^\mu \otimes \rho_{AB} \otimes \psi_{B'}^{\nu})]$, where $M_{AA'}^a$ and $N_{BB'}^b$ are unknown measurement operators defining a Positive-Operator Valued Measure (POVM). The main goal of Alice and Bob is to determine if $\rho_{AB}$ is entangled based on the knowledge of $\psi_{A'}^\mu$, $\psi_{B'}^{\nu}$, $P_{A'}(\mu)$, $P_{B'}(\nu)$, and $P(a,b|\psi_{A'}^\mu,\psi_{B'}^{\nu})$. Besides the calibration issue discussed before, this scenario is motivated by cryptographic tasks in which Alice and Bob do not trust the provider of the measurement devices they are using \cite{lo2012measurement,pirandola2015high,li2014continuous,ma2014gaussian}.

For finite dimensional Hilbert spaces, Buscemi has shown that any entangled state $\rho_{AB}$ can be certified in this scenario, but his proof is not constructive \cite{buscemi2012all}. The authors of \cite{branciard2013measurement} have shown how to construct MDI entanglement witness from standard entanglement witnesses. 
A different route was considered in \cite{vsupic2017measurement,Rosset_2018}, where the question was formulated as a convex optimisation problem that can be efficiently solved numerically.

In what follows we first generalise the results of \cite{buscemi2012all,branciard2013measurement} and show that the entanglement of every  CV entangled state can in principle be detected in a MDI scenario. 
We then move to the experimentally relevant case of Gaussian states and operations and show a MDI protocol that is able to certify the entanglement of all two-mode Gaussian entangled states. This protocol is feasible with current technology in that it only requires the production of coherent states and the implementation of homodyne measurements. Moreover, our approach provides an interesting connection between MDI entanglement detection and quantum metrology. 

\section{Reduction to process tomography}\label{sec:ProofOfPrinciple}
In this section, we show that it is possible to detect the entanglement of any entangled state in a MDI scenario where Alice and Bob use coherent states as trusted inputs (the proof is presented for two-mode bipartite states but can be generalized to $n$-modes, see~\cite{supmat}). 

Suppose Alice and Bob are in possession of trusted sources producing coherent states $\ket{\alpha}_{A'}$ and $\ket{\beta}_{B'}$, respectively, according to some distribution. The shared entangled state is  $\rho_{AB}$. The systems $AA'$ and $BB'$ are then projected onto respective two-mode squeezed vacuum (TMSV)  states, i.e. the measurement $\{\proj{\Phi^{(r)}},\eye-\proj{\Phi^{(r)}}\}$ is performed on both $AA'$ and $BB'$ ($r$ is the squeezing parameter). Conditioned on $\alpha$ and $\beta$, the probability of both measurements obtaining output `1', corresponding to the projector $\proj{\Phi^{(r)}}\equiv\Phi^{(r)}$, can be expressed as 
\begin{align}
&P_\rho(1,1|\alpha,\beta)\nonumber\\
&=\Tr\left[\left(\Phi_{AA'}^{(r)}\otimes\Phi_{BB'}^{(r)}\right)\left(\proj{\alpha}_{A'}\otimes\rho_{AB}\otimes\proj{\beta}_{B'}\right)\right]\nonumber\\
&=\Tr\left[M_{A'B'}^{(r)}\proj{\alpha}_{A'}\otimes\proj{\beta}_{B'}\right],
\label{eq:P11}
\end{align}
where we have defined
\begin{equation}\label{eq:POVM}
M_{A'B'}^{(r)}:=\Tr_{AB}\left[\left(\Phi_{AA'}^{(r)}\otimes\Phi_{BB'}^{(r)}\right)\left(\rho_{AB}\otimes\mathds{1}_{A'B'}\right)\right],
\end{equation}
which is a POVM element by construction. Our main observation in this section is that non-separability of the  POVM element defined by (\ref{eq:POVM}) is equivalent to the underlying state being entangled. Namely, we have
\begin{proposition}\label{lemma2}
For any $r>0$, the POVM element $M_{A'B'}^{(r)}$ defined by eq. (\ref{eq:POVM}) is entangled if and only if $\rho_{AB}$ is entangled.
\end{proposition}

\begin{proof}
Let us assume $\rho_{AB}$ is separable, i.e. 
\begin{equation}
\rho_{AB}=\sum_\mu p_\mu\rho^\mu_{A}\otimes\sigma^\mu_{B}.
\end{equation}
We can then define
\begin{align}
\nonumber
M^{(r)\mu}_{A'}:=&\Tr_{A}\left[\Phi_{AA'}^{(r)}\left(\rho^\mu_{A}\otimes\mathds{1}_{A'}\right)\right],\\
N^{(r)\mu}_{B'}:=&\Tr_{B}\left[\Phi^{(r)}_{BB'}\left(\sigma^\mu_{B}\otimes\mathds{1}_{B'}\right)\right],
\end{align}
which are POVM elements by construction. It is easy to see that ${
M_{A'B'}^{(r)}=\sum_\mu p_\mu M^{(r)\mu}_{A'}\otimes N^{(r)\mu}_{B'}}$, which is separable. It remains to be shown that the POVM element defined by Eq.~\eqref{eq:POVM}, which can be rewritten as \footnote{Here we used the {well-known} decomposition of two-mode squeezed states in the Fock basis
\begin{align*}
\ket{\Phi^{(r)}}= \sqrt{1-\tanh^2{r}}\; \sum_{i=0}^{\infty} (\tanh{r})^i \ket{ii}
\end{align*}
}
\begin{equation}\label{eq:entM}
 M_{A'B'}^{(r)}=(1-\lambda^2) \lambda^{\hat{n}_A+\hat{n}_B}\rho^T_{AB} \lambda^{\hat{n}_A+\hat{n}_B}
\end{equation}
(where ${\lambda=\tanh r}$ and $\hat{n}_X=a^\dagger_X a_X$ the number operator on mode $X$), is entangled for all entangled $\rho_{AB}$. In fact, suppose there exists an entanglement witness $W$ such that $\Tr[\rho W]<0$ while $\Tr[\rho' W]\geq 0$ for any $\rho'$ separable~\footnote{Such a witness always exists, as the set of separable states is defined to be closed for operational consistence, see e.g.~\cite{regula2021operational,holevo2011probabilistic}}.
From $W$ we can obtain a Hermitian operator $\tilde{W}$ such that $\Tr[M_{A'B'}^{(r)}\tilde{W}]<0$, whereas for any separable POVM $\Tr[\sum_\mu p_u (M_A^{\mu}\otimes N_B^{\mu})\tilde{W}]\geq 0$. Consider in fact 
\begin{align}
\label{eq:Wtilda}
    \tilde{W}=\lambda^{-\hat{n}_A-\hat{n}_B} W^T \lambda^{-\hat{n}_A-\hat{n}_B}\;.
\end{align}
It is then easy to see that
\begin{align}
\label{eq:POVM_EW}
\Tr[M_{A'B'}^{(r)}\tilde{W}]=(1-\lambda^2) \Tr[\rho W]<0\;.
\end{align} For separable POVMs, on the other hand, it holds
\begin{equation}
\sum_\mu p_\mu \Tr[\tilde{W}(M_A^\mu\otimes N_B^\mu)]=\sum_\mu p_\mu \Tr[W(\tilde{M}_A^\mu\otimes \tilde{N}_B^\mu)],
\end{equation}
where
$
\tilde{M}_A^\mu=\lambda^{\hat{n}_A}({M}_A^\mu)^T\lambda^{\hat{n}_A}\;,\;
\tilde{N}_B^\mu=\lambda^{\hat{n}_B}(N_B^\mu)^T\lambda^{\hat{n}_B}\;.
$
The operators $\tilde{M}_A^\mu$ and $\tilde{N}_B^\mu$ are manifestly positive semidefinite, meaning that under a proper renormalization they can be seen as states, thus generating a separable state $\rho'$ such that $\Tr[\rho' W]\geq 0$.
This implies 
\begin{align}
\sum_\mu p_\mu \Tr[W(\tilde{M}_A^\mu\otimes \tilde{M}_B^\mu)]\geq 0\;,
\end{align}
which finishes the proof.
\end{proof}

We show in~\cite{supmat} that the violation of the derived witness~\eqref{eq:POVM_EW} scales as $1/N$, where $N$ is the energy scale (number of photons) defined by the original witness, which is to be compared with the $1/d$ scaling found in~\cite{branciard2013measurement} ($d$ being the Hilbert space dimension of $\rho_{AB}$).

As a consequence of Proposition \ref{lemma2}, Alice and Bob can certify the entanglement of $\rho_{AB}$ in a MDI way, if their output statistics allow them to fully reconstruct the POVM element $M_{A'B'}^{(r)}$. As in the case of discrete variables \cite{vsupic2017measurement}, this can be achieved by means of process tomography \cite{lobino2008complete,d2001quorum}.
The set of all coherent states form a tomographically complete set via the Glauber-Sudarshan P-representation \cite{sudarshan1963equivalence,glauber1963coherent}.
Further it has been shown that 
discrete sets of coherent states can form tomographically complete sets \cite{janszky1990squeezing,janszky1995quantum}. As a special case of process tomography with coherent states, POVMs can be fully reconstructed by their output statistics \cite{lundeen2009tomography,zhang2012mapping,grandi2017experimental}. 
Once Alice and Bob have reconstructed $M^{(r)}_{A'B'}$, they can determine whether it is non-separable using an entanglement criterion. We also note that, for a given entangled state, if the witness ${W}$ is known, all that is necessary is to evaluate $\Tr[M^{(r)}_{A'B'} \tilde{W}]$, which might not require full tomography. In summary, we have the following
\begin{corollary}
For every entangled state $\rho_{AB}$, if $\ket{\alpha}$ and $\ket{\beta}$ are chosen from tomographically complete sets, Alice and Bob can certify the entanglement of $\rho_{AB}$ in a measurement-device-independent way.
\end{corollary}

The results presented in this section suffer from practical problems in their realization: firstly, they rely on performing the POVM that projects on the two mode squeezed states defined in Eq.~\eqref{eq:P11}.
A typical scheme for such measurement  involves photodetection~\cite{supmat}, which typically has low efficiency and high cost. Secondly, the full tomography could be in general experimentally inefficient. 
Therefore, the previous proof is mostly a proof-of-principle result.
Next, we show that feasible schemes for MDI entanglement detection are possible. In fact, we propose an experimentally-friendly MDI entanglement detection protocol which is based solely on homodyne measurements and can detect all two-mode Gaussian entangled states.


\section{MDI Entanglement Witness for all two-mode Gaussian states}
\label{sec:optical}

In this section we present a practical method for MDI entanglement certification of Gaussian states that can be implemented using readily available optical components. Our method is inspired by the entanglement witness introduced in a seminal paper by Duan \emph{et al.}~\cite{duan2000inseparability} (see also Simon~\cite{Simon}). In that work,
it was proven that the inequality 
\begin{align}
\label{eq:duan_witness}
\left\langle{\rm EW}_{\kappa}\right\rangle \equiv\left\langle\Delta^2 \hat{u}_\kappa\right\rangle+\left\langle\Delta^2\hat{v}_\kappa\right\rangle\geq \frac{\kappa^2+\kappa^{-2}}{2}\,,
\end{align}
where $\langle\Delta^2 \hat{O}\rangle$ is the variance of the operator $\hat{O}$, and
\begin{align}
\label{eq:uava}
 \hat{u}_\kappa=\left(\kappa\hat{x}_A-\frac{\hat{x}_B}{\kappa}\right)\,, \quad \hat{v}_\kappa=\left(\kappa\hat{p}_A+\frac{\hat{p}_B}{\kappa}\right)\,,
\end{align}
(i) holds for any two-mode separable  state, real number $\kappa$, where $\kappa>0$ without loss of generality, and pairs of  orthogonal  quadratures of the bosonic modes $A$ and $B$  
\footnote{We use here a notation similar to \cite{braunstein2005quantum}. In particular the quadratures are defined as
\begin{align*}
\hat{a}=\hat{x}+i\hat{p}\,,\quad
\hat{a}^\dagger=\hat{x}-i\hat{p}\,,
\end{align*}
such that
\begin{align*}
[\hat{x},\hat{p}]=\frac{1}{4i}[\hat{a}+\hat{a}^\dag,\hat{a}-\hat{a}^\dag]=\frac{i}{2}\,.
\end{align*}
Note for example that the variances of the quadratures will have a factor $1/2$ with respect to the normalization choice of \cite{duan2000inseparability}.}, while (ii) for any entangled Gaussian state there exist a value of $\kappa$ and pairs of quadratures such that Eq.~\eqref{eq:duan_witness} is violated. 

Our main result is an experimentally-friendly method for MDI entanglement detection inspired by the witness~\eqref{eq:duan_witness} and given by the following proposition:
\begin{proposition}
\label{prop:main}
Let $\ket\alpha$ and $\ket\beta$ be coherent states prepared by Alice and Bob according to the Gaussian probability distribution
\begin{align}
\label{eq:prior_gaus}
P(\alpha)=\frac{1}{\pi \sigma^2}e^{-|\alpha|^2/\sigma^2}
\qquad
\alpha \equiv\alpha_x+i\alpha_p\;
\end{align}
(for different choices of input distribution, see \cite{supmat}).
Consider the setup in Fig.~\eqref{fig:setup} in which uncharacterized local measurements are applied jointly on these states and half of an unknown state $\rho_{AB}$, producing as a result two real numbers $(a_1,a_2)$ for Alice and $(b_1,b_2)$ for Bob. For all local measurements and all separable states one has
\begin{align}
\left\langle{\rm MDIEW}_{\kappa}\right\rangle
\equiv\left\langle  U^2_\kappa\right\rangle+\left\langle V^2_\kappa\right\rangle 
\geq \dfrac{\kappa^2+\kappa^{-2}}{2} \dfrac{\sigma^2}{1+\sigma^2},
\label{eq:score_proposal}
\end{align}
where  $U_\kappa$ and $V_\kappa$ are
\begin{align}
\nonumber
U_\kappa\equiv\kappa a_1-\frac{b_1}{\kappa}-\frac{\kappa\alpha_x-\frac{\beta_x}{\kappa}}{\sqrt{2}}\;,
\\
V_\kappa\equiv\kappa a_2+\frac{b_2}{\kappa}-\frac{\kappa\alpha_p+\frac{\beta_p}{\kappa}}{\sqrt{2}}.
\label{eq:UkVk}
\end{align}
For any two-mode entangled Gaussian state, there exist local measurements acting jointly on the state and the input coherent states violating inequality~\eqref{eq:score_proposal}.
\end{proposition}

\begin{figure}
\includegraphics[width=0.48\textwidth]{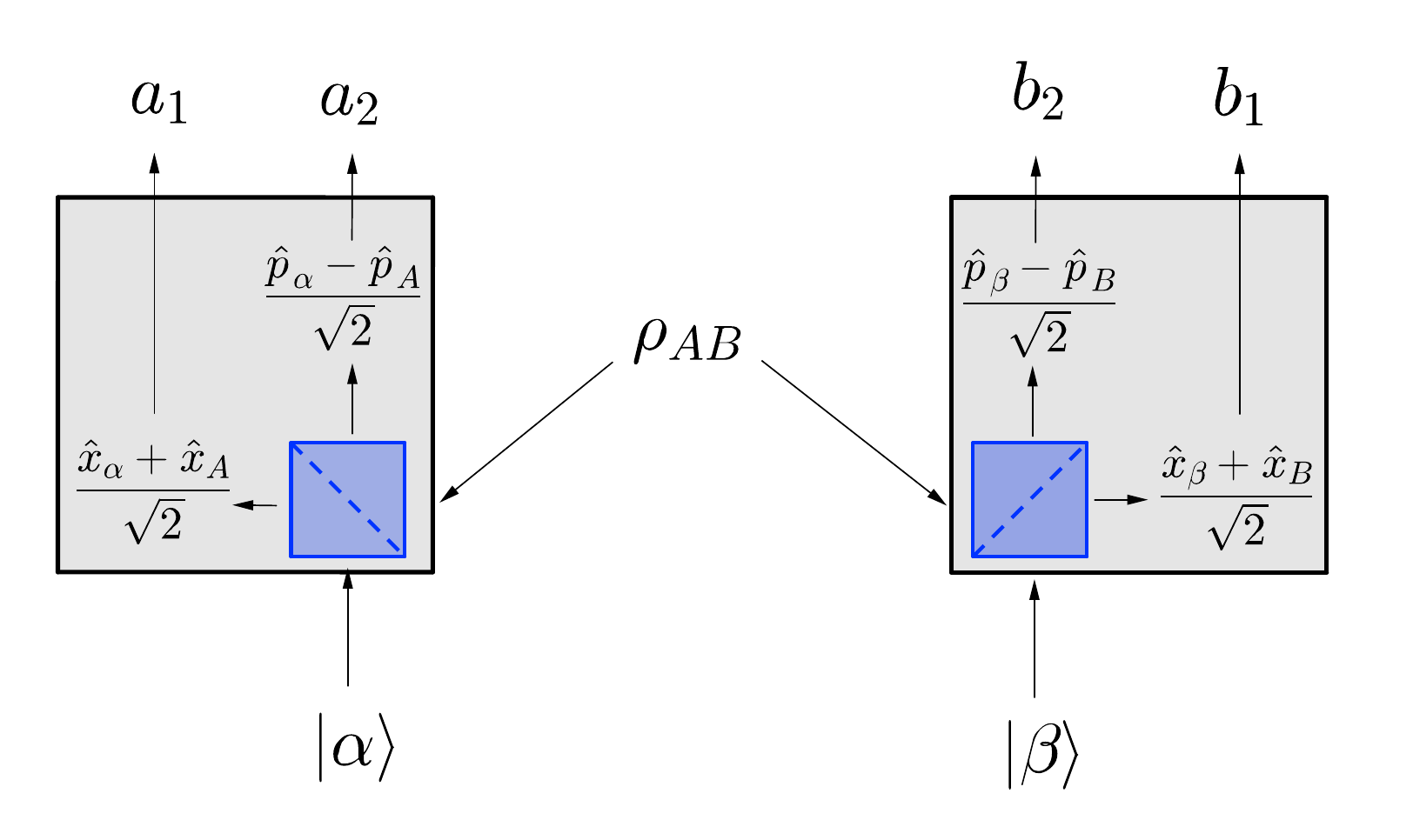}
\caption{Experimental setup for the MDI entanglement detection of a 2-mode state. Fiduciary coherent states are prepared by the parties and measured together with the corresponding subsystems of the unknown state $\rho_{AB}$. To compute the bound of the entanglement witness~\eqref{eq:score_proposal}, measurements should be seen as uncharacterized black boxes producing the outputs. To obtain a violation, the following specific measurements are implemented: the coherent states are mixed with the respective modes of $\rho_{AB}$ in a 50:50 beam splitter and homodyne measurements of $\hat{x}$ and $\hat{p}$ are performed on the two outputs.}
\label{fig:setup}
\end{figure}

Looking at its definition, the operational meaning of the witness is clear: Alice and Bob results should be such that their difference and sum, weighted by $\kappa$, are as close as possible to the same difference and sum of the quadratures of the coherent states, divided by $\sqrt 2$.
Note also that expectation values in~\eqref{eq:score_proposal} are computed with respect to the quantum state and the distribution of coherent states.


\begin{proof}[Proof of \eqref{eq:score_proposal}]
To prove the inequality we need to minimize the value of the witness over all separable states $\rho_{AB}=\sum_i p_i\rho_A^{(i)}\otimes \rho_B^{(i)}$. We can restrict the analysis to product states because the witness is linear on the state. 
It follows that the output probability factorises 
\begin{equation}
p(a_1,a_2,b_1,b_2|\alpha,\beta)=
\Tr\left[M^{A}_{a_1,a_2}\ket{\alpha}\bra{\alpha}\right]\Tr\left[M^{B}_{b_1,b_2}\ket{\beta}\bra{\beta}\right],
\end{equation}
and we are left with two independent POVMs on the input states to be optimised.
Using $\langle U_\kappa^2\rangle+\langle V_\kappa^2\rangle \geq  \langle\Delta^2 U_\kappa\rangle+ \langle\Delta^2 V_\kappa\rangle$ and the fact that the distribution is completely factorised between the two sides it follows that the value of the witness is lower bounded by the minimum of
\begin{equation}
\label{eq:metrobound}
\frac{1}{2}
\left(\kappa^2+\frac{1}{\kappa^2}\right)\left(\left\langle\Delta^2\left[\sqrt 2 a_1-\alpha_x\right]\right\rangle+ \left\langle\Delta^2\left[\sqrt 2 a_2-\alpha_p\right]\right\rangle\right).
\end{equation}
That is: in the absence of correlations, the best Alice and Bob can do is to separately perform the optimal measurements to estimate the input coherent states.

Minimizing the expression in the second parenthesis in~\eqref{eq:metrobound}  looks essentially like a metrology problem. A lower bound, in turn, can be obtained using a multi-parameter Bayesian version of the quantum Cram\'er-Rao bound \cite{yuen1973multiple}. The simultaneous estimation of the position and momentum quadratures has been studied thoroughly and is optimized for coherent states by measuring $\hat{x}$ and $\hat{p}$ on two different modes after a 50:50 beam-splitter~\cite{genoni2013optimal,morelli2020bayesian}. In particular, assuming a Gaussian prior distribution, like in our case, the minimal sum of variances is equal to $\sigma^2/(1+\sigma^2)$~\cite{genoni2013optimal}, which proves the bound for the entanglement witness.
\end{proof}

To prove the violation claimed in Proposition~\ref{prop:main}, it suffices to show that for any entangled two-mode Gaussian state there exist local measurements and values of $(\kappa, \sigma)$ that lead to it. Consider now the optical setup depicted in Fig.~\ref{fig:setup}. Alice and Bob, upon receiving their respective subsystems of $\rho_{AB}$, first mix them with local coherent states in a balanced beam splitter, and then measure the position and momentum quadratures in the output ports.  The output observables are thus
\begin{align}
\nonumber
\label{eq:obs_viol}
\hat{A}_1=\frac{\hat{x}_\alpha+\hat{x}_A}{\sqrt{2}}\;, \quad \hat{A}_2=\frac{\hat{p}_\alpha-\hat{p}_A}{\sqrt{2}}\;,
\\
\hat{B}_1=\frac{\hat{x}_\beta+\hat{x}_B}{\sqrt{2}}\;, \quad \hat{B}_2=\frac{\hat{p}_\beta-\hat{p}_B}{\sqrt{2}}.
\end{align}
The quadratures $\hat x_A$, $\hat p_A$, $\hat x_B$ and $\hat p_B$ are those used in the standard witness~\eqref{eq:duan_witness}. Observables~\eqref{eq:obs_viol} are used to define the measurement outputs needed for the computation of our MDI witness~\eqref{eq:score_proposal}. 
More precisely, consider first the case in which the average values of the state's quadratures  are null, $\langle\hat{x}_A\rangle=\langle\hat{p}_A\rangle=\langle\hat{x}_B\rangle=\langle\hat{p}_B\rangle=0$. Then, by taking $(a_1,a_2,b_1,b_2)$ equal to the statistical output of $(\hat{A}_1,\hat{A}_2,\hat{B}_1,\hat{B}_2)$ respectively, it follows by substituting in~\eqref{eq:UkVk} \footnote{We use that coherent states have minimum variances $\langle\Delta^2\hat{x}_{\alpha,\beta}\rangle=\langle\Delta^2\hat{p}_{\alpha,\beta}\rangle=\frac{1}{4}$ }
\begin{align}
\nonumber
    \left\langle U^2_\kappa\right\rangle  &=
    \left\langle\left(
    \kappa\frac{\hat{x}_\alpha+\hat{x}_A}{\sqrt{2}}-\frac{1}{\kappa}\frac{\hat{x}_\beta+\hat{x}_B}{\sqrt{2}} -\frac{\kappa \alpha_x-\frac{\beta_x}{\kappa}}{\sqrt{2}}
    \right)^2\right\rangle\\
    \nonumber
    &=\frac{1}{2}\left(
    \kappa^2\langle\Delta^2\hat{x}_\alpha\rangle+\frac{\langle\Delta^2\hat{x}_\beta\rangle}{\kappa^2} + \left\langle\Delta^2  \hat{u}_\kappa\right\rangle \right)\\
    &=\frac{1}{2}\left(\frac{\kappa^2+\kappa^{-2}}{4}+\left\langle\Delta^2  \hat{u}_\kappa\right\rangle\right)\;.
\end{align}
Similarly, $ \left\langle V^2_\kappa\right\rangle=\frac{1}{2}\left(\frac{\kappa^2+\kappa^{-2}}{4}+\left\langle\Delta^2  \hat{v}_\kappa\right\rangle\right)$, and consequently we find that in the proposed scheme
\begin{equation}
    \left\langle{\rm MDIEW}_{\kappa}\right\rangle=\frac{1}{2}\left(
    \frac{\kappa^2+\kappa^{-2}}{2}+
    \left\langle{\rm EW}_{\kappa}\right\rangle
    \right)\;.
    \label{eq:scheme_result}
\end{equation}
The generalization to states that have non-zero averages of the quadratures is obtained by simply offsetting the outputs accordingly as follows:  $a_1$ is the output of $\hat{A_1}-\langle\hat{x}_A\rangle/\sqrt{2}$, $a_2$ of $\hat{A_2}+\langle\hat{p}_A\rangle/\sqrt{2}$ and so on.

The violation of the inequality~\eqref{eq:score_proposal} is therefore found for any entangled Gaussian state: indeed, for any such state there exist quadratures and a value of $\kappa$ such that $\left\langle{\rm EW}_{\kappa}\right\rangle<  \frac{\kappa^2+\kappa^{-2}}{2}$, which implies from Eq.~\eqref{eq:scheme_result} that in the proposed scheme $ \left\langle{\rm MDIEW}_{\kappa}\right\rangle <  \frac{\kappa^2+\kappa^{-2}}{2}$. It is then sufficient to choose $\sigma$ large enough to violate~\eqref{eq:score_proposal}.

\section{Two-mode squeezed state case, with noise tolerance}
At last, to illustrate the feasibility of our scheme, we apply it to the case of $\rho_{AB}$ being a TMSV state. By including noise tolerance, we pave the way for an experimental realization of our MDI entanglement witness.

A TMSV state can be described as the mixing of two single-mode squeezed states (one squeezed in $\hat{p}$ and one in $\hat{x}$) \cite{braunstein2005quantum}. In the Heisenberg picture, this results in
\begin{align}
\nonumber
\hat{x}_A=\frac{e^r\hat{x}_1^{(0)}+e^{-r}\hat{x}_2^{(0)}}{\sqrt{2}}, \quad \hat{p}_A=\frac{e^{-r}\hat{p}_1^{(0)}+e^{r}\hat{p}_2^{(0)}}{\sqrt{2}},\\
\hat{x}_B=\frac{e^r\hat{x}_1^{(0)}-e^{-r}\hat{x}_2^{(0)}}{\sqrt{2}}, \quad \hat{p}_B=\frac{e^{-r}\hat{p}_1^{(0)}-e^{r}\hat{p}_2^{(0)}}{\sqrt{2}},
\label{eq:heisenberg_squeeze}
\end{align}
where the superscript $\{\hat{x}^{(0)},\hat{p}^{(0)}\}$ represents operators acting on the vacuum. 
Consider now the two operators
${\hat{u}_{\kappa=1}=\hat{x}_A-\hat{x}_B}$ and 
${\hat{v}_{\kappa=1}=\hat{p}_A+\hat{p}_B}$.
From Eq.~\eqref{eq:duan_witness} we see, by choosing $\kappa=1$, that these operators satisfy for any separable state
$\langle\Delta^2\hat{u}_{\kappa=1}\rangle+\langle\Delta^2\hat{v}_{\kappa=1}\rangle\geq 1$.
If we compute the above combination for the two-mode squeezed state, we obtain 
${\hat{u}_{\kappa=1}=\sqrt{2}e^{-r}\hat{x}_2^{(0)}}$
and
${\hat{v}_{\kappa=1}=\sqrt{2}e^{-r}\hat{p}_1^{(0)}}$. Consequently, it holds
\begin{align}
\langle \text{EW}_{k=1}\rangle_{\text{TMSV}}=e^{-2r}<1.
\end{align}
This is not surprising: as soon as there is squeezing $r>0$ the state is entangled.
Substituting this value into expression \eqref{eq:score_proposal}, we obtain the entangled score for the MDI witness
$ \left\langle{\rm MDIEW}_{\kappa=1}\right\rangle=\frac{1}{2}\left(1+e^{-2r}\right)$.

To check noise tolerance, we consider losses in the modes $A$ and $B$ modelled as a beam splitter
\begin{align}
\nonumber
\hat{a}_A(\eta_A)=\sqrt{1-\eta_A}\hat{a}_A(0)+\sqrt{\eta_A}\hat{a}^{(0)}_{NA}\ , \\
\hat{a}_B(\eta_B)=\sqrt{1-\eta_B}\hat{a}_B(0)+\sqrt{\eta_B}\hat{a}^{(0)}_{NB} \ ,
\label{eq:channel_noises}
\end{align}
where $\hat{a}^{(0)}_{NX}$ is a vacuum mode acting as a noise on mode $X$, while $\hat{a}_X(0)$ is the corresponding noiseless mode.
We focus on a natural scenario in which the source producing the two-mode squeezed state is between Alice and Bob and losses affect the two modes, not necessarily in a symmetric way. At the same time, the same loss noise \eqref{eq:channel_noises} applied to the input coherent states would lead only to a renormalization of $\alpha$ and $\beta$ and can be compensated by increasing the variance $\sigma$.
Using Eqs.~\eqref{eq:heisenberg_squeeze} and~\eqref{eq:channel_noises} we can accordingly compute 
\begin{widetext}
\begin{multline}
\label{eq:squeezed_score_noise}
\left\langle{\rm EW}_{\kappa}\right\rangle_{\text{TMSV},\eta_A,\eta_B}=
\left\langle \Delta^2 \left(\kappa\hat{x}_A(\eta_A)-\frac{\hat{x}_B(\eta_B)}{\kappa}\right) + \Delta^2 \left(\kappa\hat{p}_A(\eta_A)+\frac{\hat{p}_B(\eta_B)}{\kappa}\right)\right\rangle \\
=\frac{1}{2}\left(\kappa^2\eta_A+\frac{\eta_B}{\kappa^2}\right)+\frac{e^{2r}}{4}\left( \kappa\sqrt{1-\eta_A}-\frac{\sqrt{1-\eta_B}}{\kappa}\right)^2+\frac{e^{-2r}}{4}\left( \kappa\sqrt{1-\eta_A}+\frac{\sqrt{1-\eta_B}}{\kappa}\right)^2.
\end{multline}
\end{widetext}
Notice that it is always possible to choose $\kappa$ such that $ \kappa\sqrt{1-\eta_A}-\kappa^{-1}\sqrt{1-\eta_B}=0$ and $r$ big enough to nullify the last term of \eqref{eq:squeezed_score_noise}, and obtain a score~\eqref{eq:scheme_result} $\left\langle{\rm MDIEW}_{\kappa}\right\rangle=\frac{1}{4}\left(\kappa^2(\eta_A^2+1)+(\eta_B^2+1)\kappa^{-2}\right)$, which is lower than the separable bound $\frac{1}{2}\left(\kappa^2+\kappa^{-2}\right)\left(\frac{\sigma^2}{1+\sigma^2}\right)$ for large enough $\sigma$. In this sense the entanglement witness we analysed here is loss-resistant. In Figure~\ref{fig:noise_vs_entang} we plot the trade-off between noise, entanglement, and variance of the prior distributions $\sigma$ for the MDI detection of entanglement in the case of symmetric losses ($\eta_A=\eta_B$, $\kappa=1$).
\begin{figure}
\centering
\includegraphics[width=0.5\textwidth]{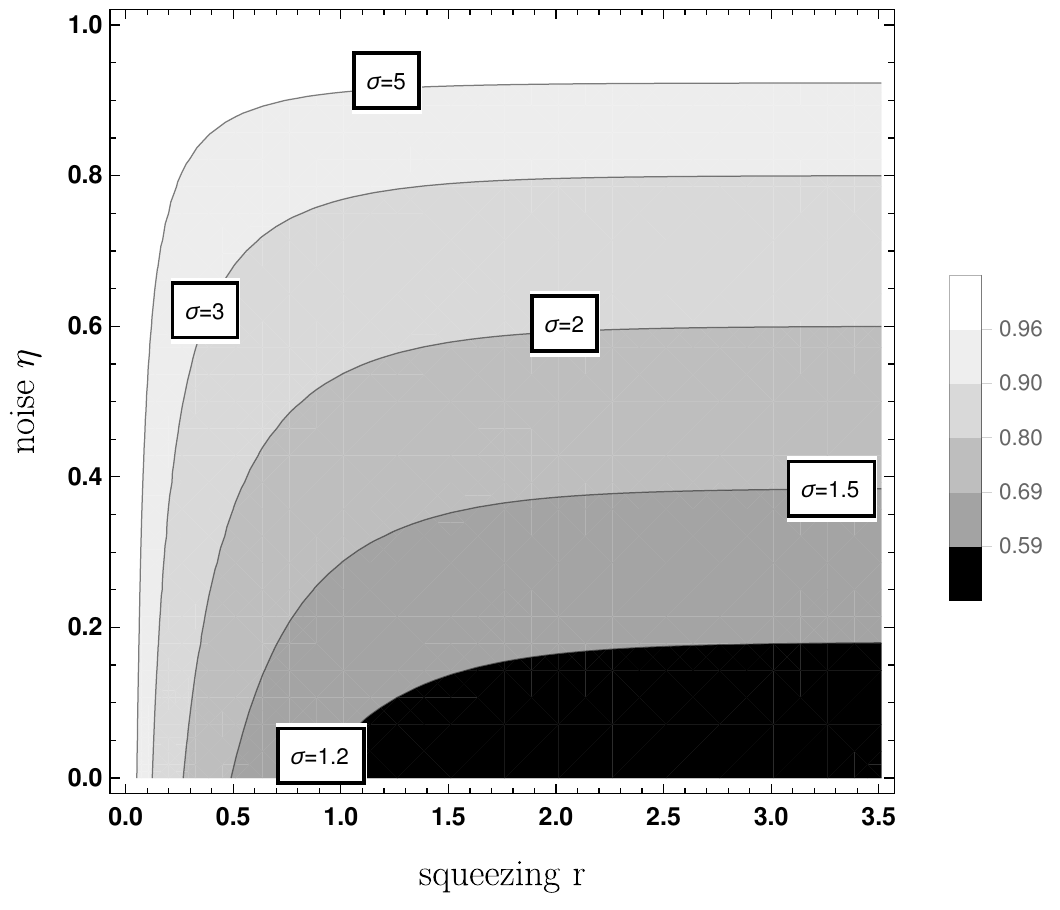}
\caption{Contourplot of the obtainable value~\eqref{eq:scheme_result} of $\left\langle{\rm MDIEW}_{\kappa=1}\right\rangle$, for a two-mode squeezed vacuum state with squeezing parameter $r$, under the presence of losses with parameter $\eta_A=\eta_B\equiv\eta$ (cf.~\eqref{eq:squeezed_score_noise}). 
The contours are chosen to match the separable bound~\eqref{eq:score_proposal} for different values of $\sigma$, which corresponds to the width of the Gaussian prior used in the experiment. Therefore
the area under each $\sigma$-contour defines the range of parameters for which MDI entanglement is certified. 
}
\label{fig:noise_vs_entang}
\end{figure}

\section{Discussion}

In this work we have promoted the task of measurement-device-independent entanglement certification to the continuous-variable regime. We first generalised the results of~\cite{buscemi2012all,branciard2013measurement} and proved that all continuous-variable entangled state can in principle be detected in this scenario. Then, we showed a simple MDI test able to detect the entanglement of all two-mode Gaussian entangled states. Most importantly, the test only relies on the preparation of coherent states and uses standard experimental setups, thus being readily available with current technology. 

Our work also opens up a series of interesting directions. From a general perspective, our works opens the path to the use of CV quantum systems for MDI tasks beyond entanglement detection, such as randomness generation or secure communication. Another possible research direction would be to investigate the possible generality of the connection between entanglement detection and metrology exploited here. In particular, can all MDI entanglement tests be translated into a parameter estimation problem? We are therefore confident that our results will motivate further studies in the field of quantum information with continuous variables.

\begin{acknowledgments}


 We thank Roope Uola for insightful discussions.
 This work was supported by the Government of Spain (FIS2020-TRANQI and Severo Ochoa CEX2019-000910-S), Fundaci\'{o} Cellex, Fundaci\'{o} Mir-Puig, Generalitat de Catalunya (CERCA, AGAUR SGR 1381 and QuantumCAT). A.A. is supported by the ERC AdG CERQUTE and the AXA Chair in Quantum Information Science. P. A. is supported by ``la Caixa" Foundation (ID 100010434, Grant  No.  LCF/BQ/DI19/11730023). S.B. acknowledges funding from the European Union's Horizon 2020 research and innovation program, grant agreement No. 820466 (project CiViQ). D.C. is supported by a Ramon y Cajal Fellowship (Spain).

\end{acknowledgments}

\medskip
\bibliography{MDIEW.bib}

\onecolumngrid
\appendix

\section{On the feasibility of Proposition~\ref{lemma2}}
\label{app:realising_prop1}

The proof of principle realisation of a MDI entanglement witness proposed in Sec.~\ref{sec:ProofOfPrinciple} is based on the possibility of performing the POVM defined by the projection on the two-mode squeezed vacuum state $\Phi^{(r)}$. Here we show that in principle this is realisable using photodetectors, single mode squeezing and a beam splitter. Indeed, it is known that a two-mode squeezed vacuum can be obtained by applying the following unitaries on a two-mode vacuum \cite{braunstein2005quantum}:
\begin{align}
    \ket{\Phi^{(r)}}_{12}=\hat{B}_{12}\hat{S}^{(r)}_1 \hat{S}^{(-r)}_2 \ket{00}_{12}\ .
\end{align}
Here $\hat{S}_X^{(r)}$represents a single-mode squeezing operation with squeezing parameter $r$ on mode $X$, which satisfies $\hat{S}_X^{(r)\dagger}=\hat{S}_X^{(-r)}$. The following unitary $\hat{B}_{12}$ is induced by a 50:50 beamsplitter. Therefore, for any two mode state $\ket{\psi}_{12}$, the projection on $\Phi^{(r)}$ can be simulated by a photodetection preceded by the corresponding inverse unitary
\begin{align}
_{12}\bra{\Phi^{(r)}}\psi\rangle_{12}=
_{12}\bra{00}\hat{S}^{(-r)}_1 \hat{S}^{(r)}_2\hat{B}^\dagger_{12}\ket{\psi}_{12}\;.
\end{align}
Additionally, the result of a projection on the vacuum might in principle simulated by double homodyne measurement of $\hat{x}$ and $\hat{p}$ on the incoming mode mixed with a vacuum mode in a 50:50 beam splitter. Indeed the statistics of such a continuous measurement (also called heterodyne measurement) is described by the Husimi Q-function \cite{genoni2013optimal,adesso2014continuous}, 
\begin{align}
    P\left(\frac{x}{\sqrt{2}},\frac{p}{\sqrt{2}}\right)dxdp=\frac{\bra{x+ip}\rho\ket{x+ip}}{\pi} dxdp\;,
\end{align}
therefore the projection $\bra{0}\rho\ket{0}$ can be approximated by counting the statistical outputs of $(x,p)$ being close to $(0,0)$.
The drawback of such a method of simulating photodetection would be discarding all the measurement that are outside the confidence-interval approximating $0$ on all the modes.

We notice as well that the energy scale defined by the original entanglement witness $W$ of $\rho_{AB}$ sets a bound on the violation of the POVM-entanglement-witness \eqref{eq:POVM_EW} proposed in the main text. Recall that
${\tilde{W}=\lambda^{-\hat{n}_A-\hat{n}_B} W \lambda^{-\hat{n}_A-\hat{n}_B}}$, where $\hat{n}_X=a^\dagger_X a_X$ is the photon-number operator of mode $X$. Therefore,
for $\tilde{W}=$ to be bounded,  $W$ needs to have finite energy as well. Suppose $W\lesssim e^{-(n_A+n_B)/N}$ for large energies. The value of $N$ is an upperbound on the energy scale of $W$. Then, from \eqref{eq:Wtilda} we see that for $\tilde{W}$ to be bounded as well, it is necessary that $\lambda^{-1}e^{-1/N}<1$ (recall that $0<\lambda<1$), hence good values for $\lambda$ are
\begin{align}
    e^{-1/N}<\lambda<1\;,
\end{align}
and so 
\begin{align}
 1-e^{-2/N}   > 1-\lambda^2 > 0\ .
\end{align}
Confronting this last bound with \eqref{eq:POVM_EW} we notice that for large energies $\mathcal{O}(N)$ of the original witness, the proposed violation is dampened by a factor $\sim\frac{2}{N}$.

\section{Generalization of Proposition~\ref{lemma2} to N modes}
\label{app:general_Nmodes}

Proposition~\ref{lemma2} can be easily generalised to states having any number of modes. To be explicit how this can be done, we show the generalisation to a bipartite state having two modes on Alice side, $A_1$, $A_2$, and one mode on Bob side, $B$. In such a case we substitute Eq.~\eqref{eq:P11} by
\begin{align}
P_\rho(1,1,1|\alpha_1,\alpha_2,\beta)
&=\Tr\left[\left(\Phi_{A_1A'_1}^{(r)}\otimes\Phi_{A_2A'_2}^{(r)}\otimes\Phi_{BB'}^{(r)}\right)\left(\proj{\alpha_1}_{A'_1}\otimes\proj{\alpha_2}_{A'_2}\otimes\rho_{A_1A_2B}\otimes\proj{\beta}_{B'}\right)\right]\nonumber\\
&=\Tr\left[M_{A'_1A'_2B'}^{(r)}\proj{\alpha_1}_{A'_1}\otimes\proj{\alpha_2}_{A'_2}\otimes\proj{\beta}_{B'}\right],
\label{eq:P11app}
\end{align}
where it is defined
\begin{equation}\label{eq:POVMapp}
M_{A'_1A'_2B'}^{(r)}:=\Tr_{A_1A_2B}\left[\left(\Phi_{A_1A'_1}^{(r)}\otimes\Phi_{A_2A'_2}^{(r)}\otimes\Phi_{BB'}^{(r)}\right)\left(\rho_{A_1A_2B}\otimes\eye_{A'_1A'_2B}\right)\right]\;.
\end{equation}
That is, for the generalization, each party generates a number of coherent states equal to the number of modes of his side of the partition;  projections on TMSV states are performed accordingly.
The derivation then follows in the same way as presented in the main text, that is, noticing that
\begin{equation}\label{eq:entMapp}
 M_{A'_1A'_2B'}^{(r)}=(1-\lambda^2)^{\frac{3}{2}} \lambda^{\hat{n}_{A_1}+\hat{n}_{A_2}+\hat{n}_B}\rho_{A_1A_2B}^T \lambda^{\hat{n}_{A_1}+\hat{n}_{A_2}+\hat{n}_B},
\end{equation}
and defining the entanglement witness $\tilde{W}$ for $M_{A'_1A'_2B'}^{(r)}$ in terms of the original entanglement witness $W$ of $\rho_{A_1A_2B}$,
\begin{align}
\label{eq:Wtildaapp}
    \tilde{W}=\lambda^{-\hat{n}_{A_1}-\hat{n}_{A_2}-\hat{n}_B} W^T \lambda^{-\hat{n}_{A_1}-\hat{n}_{A_2}-\hat{n}_B}\;.
\end{align}
It follows that
\begin{align}
\label{eq:POVM_EWapp}
\Tr[M_{A'_1A'_2B'}^{(r)}\tilde{W}]=(1-\lambda^2)^\frac{3}{2} \Tr[\rho_{A_1A_2B} W]<0\;,
\end{align} 
while for separable POVMs, it holds
\begin{equation}
\sum_\mu p_\mu \Tr[\tilde{W}(M_{A_1A_2}^\mu\otimes N_B^\mu)]=\sum_\mu p_\mu \Tr[W(\tilde{M}_{A_1A_2}^\mu\otimes \tilde{N}_B^\mu)],
\end{equation}
where
$
\tilde{M}_{A_1A_2}^\mu=\lambda^{\hat{n}_{A_1}+\hat{n}_{A_2}}(M_{A_1A_2}^\mu)^T\lambda^{\hat{n}_{A_1}+
\hat{n}_{A_2}}\;,\;
\tilde{N}_B^\mu=\lambda^{\hat{n}_B}(N_B^\mu)^T\lambda^{\hat{n}_B}\;.
$
Following the same reasoning as in the main text, the operators $\tilde{M}_{A_1A_2}^\mu$ and $\tilde{N}_B^\mu$ are positive, and can be seen as representing unnormalized states, thus generating a separable state $\rho'$ such that $\Tr[\rho' W]\geq 0$.

Hence we see that the proof of Proposition~\ref{lemma2} works for the case in which Alice has two modes.
The generalisation to any number of modes is straightforward.

\section{Fisher information matrix of the prior distribution, and possible prior choices.}
The separable lower bound~\eqref{eq:score_proposal} of Proposition~\ref{prop:main} in the main text, is derived from the multi-parameter quantum Cram\'er Rao bound based on the right logarithmic derivative (RLD) \cite{holevo2011probabilistic}, in its Bayesian form with additional prior information \cite{yuen1973multiple}. Such bound, when applied to the minimization of the sum of the single parameter variances (in our case $\alpha_x$ and $\alpha_p$, cf. Eq.~\eqref{eq:metrobound}), is presented in Eq.(14) of Ref.\cite{genoni2013optimal}. The information associated to the prior distribution $P(\alpha_x,\alpha_p)$ is encoded in the prior's Fisher Information Matrix (FIM) defined in Eq.(10) of the same Ref.\cite{genoni2013optimal}, that is
\begin{align}
    A_{ij}:=\int d\alpha_x d\alpha_p\; P(\alpha_x,\alpha_p)\frac{\partial \log P(\alpha_x,\alpha_p)}{\partial \alpha_i} \frac{\partial \log P(\alpha_x,\alpha_p)}{\partial \alpha_j}\;,
\end{align}
where $i$ and $j$ can be equal to $x$ and $p$. For the Gaussian distribution used in the main text \eqref{eq:prior_gaus} $P(\alpha_x,\alpha_p)=\frac{1}{\pi\sigma^2}\exp{[-(\alpha_x^2+\alpha_p^2)/\sigma^2]}$ it is not difficult to obtain that $A$ is diagonal
\begin{align}
\label{eq:gauss_fim}
    A^{\rm Gauss}=\begin{pmatrix}
    \frac{2}{\sigma^2} & 0 \\
    0 & \frac{2}{\sigma^2}
    \end{pmatrix}\;.
\end{align}
The property of $A$ being diagonal holds for any distribution in the form $P(\alpha_x,\alpha_p)=P(\alpha_x)P(\alpha_p)$ which is symmetric. In such a case $A_{ij}=F_{i}\delta_{ij}$, where $F_i$ is the single parameter Fisher information relative to $P(\alpha_i)$,
\begin{align}
    A^{\rm symm}=
    \begin{pmatrix}
    F_x & 0 \\
    0 & F_p
    \end{pmatrix}\;,\qquad F_i=\int d\alpha_i\; P(\alpha_i)\frac{\partial \log P(\alpha_i)}{\partial \alpha_i}\;.
\end{align}

As an example, it is possible to choose an almost-flat distribution on a square of size $l$, defined as
\begin{align}
    P(\alpha_x,\alpha_p)=\frac{1}{l^2}I_{\delta,l}(\alpha_x)I_{\delta,l}(\alpha_p)
\end{align}
where $I_{\delta,l}$ is a smooth version of the indicator function on the interval $[-l/2,+l/2]$, 
\begin{align}
  I_{\delta,l}=\begin{cases}
  0 & x\leq -\frac{l}{2}-\frac{\delta}{2} \\
  \frac{1}{2}+\frac{1}{2}\sin(\frac{\pi}{\delta}(x+\frac{l}{2})) & -\frac{l}{2}-\frac{\delta}{2}\leq x \leq -\frac{l}{2}+\frac{\delta}{2}\\
  1 & -\frac{l}{2}+\frac{\delta}{2}\leq x \leq \frac{l}{2}-\frac{\delta}{2}\\
   \frac{1}{2}-\frac{1}{2}\sin(\frac{\pi}{\delta}(x-\frac{l}{2})) & \frac{l}{2}-\frac{\delta}{2}\leq x \leq \frac{l}{2}+\frac{\delta}{2}\\
   0 & x\geq \frac{l}{2}+\frac{\delta}{2}
  \end{cases}
\end{align}
Note that for consistency it is needed $\delta\leq l$, and the limit cases $\delta=0$ and $\delta=l$, correspond respectively to $I_{\delta,l}$ being the indicator function and $I_{\delta,l}$ being a single symmetric cosinusoidal wave between $-l$ and $+l$.
For such a choice of $P(\alpha_x,\alpha_p)$ the corresponding fisher information matrix is easily computed as
\begin{align}
    A^{(\delta,l)}=
    \begin{pmatrix}
    \frac{\pi^2}{l\delta} & 0 \\
    0 & \frac{\pi^2}{l\delta}
    \end{pmatrix}\;.
\end{align}
We see that in the limit $\delta\rightarrow 0$ the FIM diverges and becomes useless for computing Cram\'er Rao bounds, which become trivial in such limit. This is related to the fact that the multi-parameter Cram\'er Rao bounds are not tight in general \cite{holevo2011probabilistic} (notice that in case of Gaussian prior it can be saturated \cite{genoni2013optimal}). In the limit $\delta=l$, we see that such a prior distibution is equivalent to the Gaussian case~\eqref{eq:gauss_fim}, provided
\begin{align}
   \frac{l}{\pi}=\frac{\sigma}{\sqrt{2}} \;.
\end{align}


\end{document}